\documentclass[journal]{IEEEtran}
\usepackage[latin1]{inputenc}
\usepackage{graphicx}
\usepackage{amssymb,amsmath,amsfonts,amsthm}
\usepackage{float}
\usepackage{graphics}
\usepackage{xspace}
\usepackage[usenames,dvipsnames]{color}
\usepackage{amssymb, epsfig, hyperref}
\usepackage{amsmath}
\usepackage{physics}
\usepackage{epstopdf}
\newtheorem{proposition}{Proposition}

\usepackage[ruled,vlined,linesnumbered,noend]{algorithm2e}
\SetKwInput{KwIn}{Input}\SetKwInput{KwOut}{Output}

\begin{document}

\title{\LARGE Quantum Optical Integrated Sensing and Communication with Homodyne BPSK Detection}

\author{Ioannis Krikidis,~\IEEEmembership{Fellow, IEEE}
\thanks{I. Krikidis is with the Department of Electrical and Computer Engineering, University of Cyprus, Cyprus (e-mail: krikidis@ucy.ac.cy).}
\thanks{This work was supported by the ERC (Horizon Europe, Grant No. 101241675, PoC QUARTO) and the EU Horizon JU-SNS (Grant No. 101139291, iSEE-6G).}}

\maketitle

\begin{abstract}
In this letter, we propose a quantum integrated sensing and communication scheme for a quantum optical link using binary phase-shift keying modulation and homodyne detection. The link operates over a phase-insensitive Gaussian channel with an unknown deterministic phase rotation, where the homodyne receiver jointly carries out symbol detection and phase estimation. We formulate a design problem that minimizes the bit-error rate subject to a Fisher information-based constraint on estimation accuracy. To solve it, we develop an iterative algorithm composed of an inner expectation-maximization loop for joint detection and estimation and an outer loop that adaptively retunes the local oscillator phase. Numerical results confirm the effectiveness of the proposed approach and demonstrate a fundamental trade-off between communication reliability and sensing accuracy.
\end{abstract}

\vspace{-0.2cm}

\vspace{-0.3cm}
\section{Introduction}
\IEEEPARstart{I}{ntegrated} sensing and communication (ISAC) is envisioned as a key enabler of sixth-generation wireless systems, where communication signals are co-designed not only to transmit information but also to sense and interact with the surrounding environment \cite{ZHA}. In recent years, ISAC has been extensively studied from multiple perspectives, including information-theoretic analysis, physical-layer techniques, networking aspects, and practical implementations \cite{MAS}. A key challenge in ISAC is the fundamental trade-off between communication and sensing performance, since a single system architecture cannot simultaneously achieve optimal efficiency for both objectives.

On the other hand, quantum communication is a promising paradigm that aims to achieve non-conventional benefits by exploiting fundamental principles of quantum mechanics, such as entanglement, linear superposition, and quantum measurement \cite{DJO}. In particular, quantum optical communication models light as a quantum harmonic oscillator and embeds classical information into quantum states. Recently, this communication technology has attracted significant attention, primarily due to its applications in quantum key distribution, but also from a broader communication-theoretic perspective. For instance, the work in \cite{JUN} investigates coding techniques for coherent-state-based quantum optical communications over multipath channels, while the work in \cite{KRI} designs a batteryless quantum optical link and studies its fundamental performance limits. In addition, quantum optical signals can enhance the sensitivity of parameter estimation in sensing applications by exploiting quantum-mechanical features, a mature research field with major applications in quantum radar systems \cite{LAZ}. Extending the ISAC concept to quantum optical communication systems gives rise to quantum ISAC (QISAC), a new paradigm with promising applications \cite{CON}. Compared to classical ISAC systems that balance rate, power, and sensing accuracy under classical signal and noise assumptions, the QISAC framework operates under quantum noise and measurement uncertainty, exposing analogous but fundamentally constrained trade-offs imposed by quantum physics.

In this letter, we present for the first time in the literature a practical QISAC model for an optical quantum link, where a homodyne receiver jointly detects binary phase-shift keying (BPSK)-modulated signals and estimates (senses) a deterministic channel phase rotation. The QISAC design problem is formulated as an optimization task that minimizes the bit-error rate (BER) of BPSK detection subject to a prescribed estimation quality, expressed through the block Fisher information with respect to the local oscillator (LO) phase of the homodyne receiver. To solve this problem, we develop an iterative algorithm consisting of (i) an inner expectation-maximization (EM) loop that jointly detects the transmitted symbols and estimates the unknown channel phase, and (ii) an outer loop that updates the LO phase by gradually steering it toward communication- or sensing-optimal angles based on a Fisher-information feasibility constraint. Numerical results validate the efficiency of the proposed scheme and reveal the fundamental trade-off between communication reliability and sensing accuracy. Our work does not aim to demonstrate a quantum advantage over classical ISAC systems, but to establish a physically consistent quantum-optical formulation of the ISAC problem based on coherent states and homodyne detection, thereby bridging classical estimation methods with quantum-optical modeling.

\noindent {\it Notation:} $|x\rangle$ (called a ket) denotes a vector in a complex Hilbert space, $\langle x|Y|x\rangle$ represents the expectation value of operator $Y$ when the system is in state $|x\rangle$, $[A,B]=AB-BA$ denotes the commutator for two operators $A$ and $B$; $\mathcal{N}(x;\mu,\sigma^2)$ denotes the Gaussian probability density function (pdf) with mean $\mu$ and variance $\sigma^2$, $Q(\cdot)$ denotes the Q-function, and $\Tr(\cdot)$ represents the trace operator.
 
\vspace{-0.3cm}
\section{System Model}\label{sysec}

We assume a simple quantum optical link where the transmitter encodes classical information into BPSK coherent states. Specifically, for each channel use the transmitter prepares a coherent state \cite{JUN}
\vspace{-0.1cm}
\begin{equation}
	|\alpha_m\rangle = |\sqrt{E}\, e^{i\varphi_m}\rangle,\;\;\textrm{with}\;\;
	\varphi_m =\pi m,\;\; m=0,1,
\end{equation}
where $E$ denotes the average photon number per symbol (the symbol energy) and each symbol is equiprobable. 

We model the optical link as a phase-insensitive Gaussian channel with transmissivity $\eta \in (0,1]$, thermal noise of mean photon number $N_a$, and an unknown quasi-static phase rotation $\theta$ induced by medium imperfections such as path fluctuations and laser phase noise; $\theta$ may also capture environmental interactions, {\it e.g.}, target reflections. This channel, characterized by fixed transmissivity and constant phase rotation, is a fundamental model for coherent optical links and provides a realistic description of line-of-sight quantum communication systems \cite{Holevo}.

For a BPSK coherent input $|\alpha_m\rangle$, the channel output is a displaced thermal state \cite{DJO}, which can be written as a Gaussian mixture over all coherent states (density operator in P-representation) {\it i.e.,}
\vspace{-0.1cm}
\begin{equation}
\rho_{\mathrm{out}}(m)
= \frac{1}{\pi N_a}\!\int
\exp\!\left(-\frac{|\alpha-\beta_m|^2}{N_a}\right)\,
|\alpha\rangle\langle\alpha|\,\mathrm{d}^2\alpha,
\end{equation}
where $\beta_m=\sqrt{\eta}\,\alpha_m\,e^{i\theta}$ is the mean field after attenuation and phase rotation, and $|\alpha\rangle$ denotes a coherent state with eigenvalue $\alpha$.

At the receiver, we employ a \emph{homodyne detector} that measures the field quadrature aligned with a LO of phase $\psi$. The quadrature operator is defined as \cite[Ch. 5.8]{DJO}
\vspace{-0.1cm}
\begin{equation}
    X_\psi = \frac{1}{\sqrt{2}}\left(a\,e^{-i\psi}+a^\dagger e^{i\psi}\right),
\end{equation}
where $a$ ($a^\dagger$) is the annihilation (creation) operator of the received mode with commutator $[a,a^\dagger]=1$ and $a|\alpha\rangle=\alpha|\alpha\rangle$. 

For the displaced thermal state $\rho_{\mathrm{out}}(m)$, the homodyne measurement outcome is a real random variable with Gaussian distribution $\mathcal{N}\!\big(x;\mu_m, \sigma^2\big)$
with mean and variance given, respectively, by 
\vspace{-0.1cm}
\begin{align}
\mu_m&=\langle X_\psi \rangle = \mathrm{Tr}\big(\rho_{\mathrm{out}}(m)\,X_{\psi}\big)= A\cos(\varphi_m+\phi), \\
\sigma^2&=\langle X_\psi^2\rangle - \langle X_\psi\rangle^2= N_a + \tfrac{1}{2},
\end{align}
where $A=\sqrt{2\eta E}$, and $\phi \triangleq \theta - \psi$ denotes the \emph{effective phase offset}, {\it i.e.}, the mismatch between the unknown channel phase $\theta$ and the controllable LO phase $\psi$.

\begin{figure}
	\centering
	\includegraphics[width=0.67\linewidth]{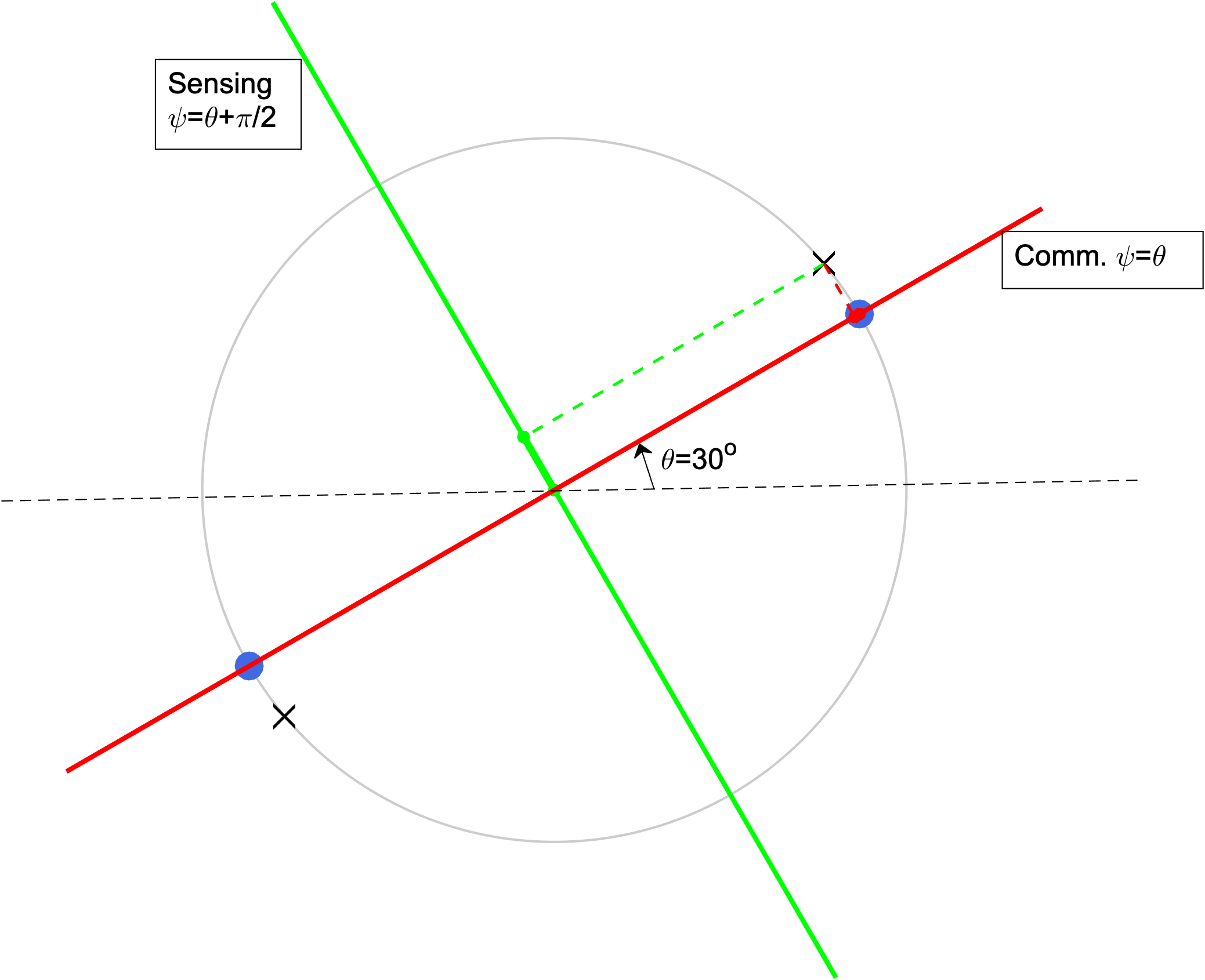}
	\vspace{-0.2cm}
	\caption{Communication-sensing trade-off for BPSK with channel 
		rotation $\theta=30^\circ$ and perturbation $\delta\theta=10^\circ$. 
		The red axis ($\psi=\theta$) maximizes symbol separation and is 
		communication-optimal, while the green axis ($\psi=\theta+\pi/2$) 
		maximizes the projection shift due to $\delta\theta$ and is sensing-optimal. 
		Dashed lines indicate the projections of the perturbed symbols onto each axis.}
	\label{system}
\end{figure}

\vspace{-0.3cm}
\section{Quantum integrated sensing and communication}

We consider a simple QISAC scenario where the quantum optical transmitter sends a block of $N$ BPSK symbols (see Sec.~\ref{sysec}) and the receiver jointly performs communication and sensing tasks. Specifically, the receiver adjusts the LO phase $\psi$ in order to jointly (i) demodulate the transmitted symbols with low probability of error, and (ii) estimate the unknown channel phase rotation $\theta$ with sufficient accuracy.  

Formally, the QISAC design problem is formulated as the minimization of the BER, subject to a constraint on the estimation accuracy of the unknown parameter $\theta$, quantified by the Fisher information (or, equivalently, by the Cram\'er--Rao bound\footnote{By the Cram\'er--Rao bound, any unbiased estimator of $\theta$ satisfies 
$\mathrm{Var}(\hat{\theta}) \ge 1/F_c(\psi,\theta)$. Thus, enforcing 
$F_c(\psi,\theta) \ge \Gamma_{\min}$ guarantees that the estimator variance 
is at most $1/\Gamma_{\min}$; larger $\Gamma_{\min}$ means lower variance (better estimation performance).}). Specifically, the optimization problem considered is given by
\vspace{-0.1cm}
\begin{align}
[P1]\;\;&\min_{\psi \in [0,\pi] }\;  P_e(\psi,\theta) \nonumber \\
&\text{s.t.}\quad  F_c(\psi,\theta)\;\geq\;\Gamma_{\min},
\end{align}
where $P_e(\psi,\theta)$ denotes the BER for BPSK as a function of the LO phase $\psi$ and the channel rotation $\theta$, $F_c(\psi,\theta)=NF(\psi,\theta)$ denotes the block Fisher information associated with $\theta$, $F(\psi,\theta)$ is the symbol Fisher information, and $\Gamma_{\min}$ represents the required minimum estimation quality with $\Gamma_{\min}\leq F_c^{\max}\leq NA^2/\sigma^2$ (see Appendix \ref{ap2}).

Since the true channel parameter $\theta$ is unknown at the receiver, it is 
estimated from a block of $N$ homodyne measurements ({\it e.g.}, using the proposed EM algorithm). 
The resulting estimate $\hat{\theta}$ is then substituted into the expressions for $P_e(\cdot)$ 
and $F_c(\cdot)$ when solving problem [P1] with respect to $\psi$. 
Before introducing the proposed QISAC algorithm, we first derive closed-form analytical 
expressions for the BER $P_e(\cdot)$ and the Fisher information $F_c(\cdot)$, 
which are summarized in the following proposition.
\vspace{-0.1cm}
\begin{proposition}
For BPSK transmitted coherent states and a homodyne receiver, 
the BER and the symbol Fisher information for the 
deterministic parameter $\theta$ are given by
\begin{align}
&P_e(\psi,\theta)=Q\left(\frac{A}{\sigma}\!\left|\cos(\phi)\right|\right), \label{error}\\
&F(\psi, \theta)\!=\!\!\frac{1}{2}\!\!\int_{-\infty}^{\infty}\!\!\!
\frac{\Big(\!\!\sum_{m=0}^{1}\mathcal{N}\!\big(x;\mu_m,\sigma^2\big)\,
\frac{(x-\mu_m)\,\mu'_m}{\sigma^2}\!\Big)^2}
{\sum_{m=0}^{1}\mathcal{N}\!\big(x;\mu_m,\sigma^2\big)}\,\!dx, \label{fish}
\vspace{-0.8cm}
\end{align}
where $\mu'_m=-A\sin(\varphi_m+\phi)$. We also note that $F_c^{\max}=N\max_{\phi}F(\psi,\theta)$ which can be computed numerically (1-D search). 
\end{proposition}
\begin{proof}
The proof is given in Appendix \ref{ap1}. 
\end{proof}
In addition, a careful observation of the BER in~\eqref{error} and the Fisher information 
in~\eqref{fish} reveals that the LO phase $\psi$ induces a fundamental trade-off between 
communication and sensing. Specifically, aligning the LO phase $\psi$ toward the communication-optimal direction improves demodulation performance but reduces phase sensitivity, 
whereas aligning it toward the sensing-optimal direction enhances estimation accuracy at the cost of a higher BER. In the following proposition, we present the optimal LO phases $\psi$ that facilitate communication and sensing tasks, highlighting their inherent conflict. The expressions are derived under a simplified high- signal-to-noise ratio (SNR) assumption and form the basis for the outer loop of the proposed QISAC algorithm. Fig. \ref{system} schematically also presents this critical observation. 
\vspace{-0.1cm}
\begin{proposition}\label{p2}
\label{prop:angles}
For BPSK transmitted coherent states and a homodyne receiver in the high-SNR 
regime, (i) the LO phase that minimizes the BER is $\psi_{\mathrm{com}} = \theta+k\pi$, and (ii) the LO phase that maximizes the Fisher information is $\psi_{\mathrm{sen}} = \theta + \frac{\pi}{2}(2k+1)$, where $k\in \mathbb{Z}$. Therefore, the communication- and sensing-optimal LO phases differ by exactly 
$\pi/2$.
\end{proposition}
\begin{proof}
	The proof is given in Appendix \ref{ap2}. 
\end{proof}

The proposed QISAC algorithm follows a two-loop structure {\it i.e.,} in the inner loop, 
given a current value of the LO phase~$\psi$, we apply the EM algorithm 
to jointly estimate the unknown channel phase~$\theta$ and detect the 
transmitted symbols (latent variables). In the outer loop, the estimate of~$\theta$ is used 
to re-tune~$\psi$ by exploiting the remarks in Proposition \ref{p2}; the EM step is then repeated with the updated~$\psi$, and the process iterates until convergence\footnote{In scenarios where the phase $\theta$ varies within a block (due to laser phase noise, drift, or turbulence) the proposed framework could be extended with recursive or Kalman-based phase-tracking mechanisms, enabling real-time adaptation and improved robustness to rapid phase fluctuations.}. In the following discussion, we present the two loops of the proposed algorithm. 

\subsubsection{Inner loop (estimate $\theta$ and demodulate BPSK symbols for a given $\psi$)}
We formulate QISAC as a maximum-likelihood (ML) estimation problem in which the transmitted 
BPSK symbols are latent variables. Let $s_n \in \{0,1\}$ denote the 
latent variable corresponding to the transmitted symbol index in the 
$n$-th channel use, and let $\{x_n\}_{n=1}^N$ denote the homodyne 
observations. Consequently, an EM scheme 
can be employed, alternating between computing the symbol posteriors 
(E-step) and maximizing the expected complete-data log-likelihood over 
$\theta$ (M-step), thereby yielding joint estimates of $\theta$ and the 
symbol sequence. Specifically, the EM algorithm starts with an initial 
guess $\theta^{(0)}$ ({\it e.g.}, random initialization) and then iterates 
between the following two steps until convergence.

\textbf{E-step (symbol posteriors):} For each observation $x_n$ and each symbol index 
$m \in \{0,1\}$, we compute the posterior probability
\vspace{-0.3cm}
\begin{align}
\gamma_{n,m}^{(t)} &= \Pr(s_n=m \mid x_n,\theta^{(t)}) \nonumber \\
&= \frac{\mathcal{N}\!\big(x_n;\mu_m(\psi,\theta^{(t)}),\sigma^2\big)}
{\sum_{k=0}^{1}\mathcal{N}\!\big(x_n;\mu_k(\psi,\theta^{(t)}),\sigma^2\big)}.
\end{align}
These responsibilities represent the probability that observation $x_n$ 
was generated by symbol index $m$.

\textbf{M-step (parameter update):} We update $\theta$ by maximizing the 
expected complete-data log-likelihood
\vspace{-0.3cm}
\begin{align}
\mathcal{Q}(\theta \mid \theta^{(t)}) 
&= \mathbb{E}_{s_n\mid x_n,\,\theta^{(t)}} 
\Big[ \log p(x_n,s_n\mid \psi,\theta) \Big] \nonumber \\[1ex]
&= \sum_{n=1}^{N}\sum_{m=0}^{1}\gamma_{n,m}^{(t)}\,
\log \mathcal{N}\!\big(x_n;\mu_m(\psi,\theta),\sigma^2\big).
\end{align}
Equivalently, this step can be written as the weighted quadratic problem 
\vspace{-0.1cm}
\begin{align}
\theta^{(t+1)}
&=\arg\max_{\theta}\;\mathcal{Q}(\theta \mid \theta^{(t)})  \nonumber \\
&=\arg\min_{\theta}\;\sum_{n=1}^{N}\sum_{m=0}^{1}\gamma_{n,m}^{(t)}\,
\big(x_n-\mu_m(\psi,\theta)\big)^2. \label{mstep}
\end{align}
The M-step in \eqref{mstep} does not admit a general closed-form solution for $\theta$.
We can solve it numerically via Newton's iterative method by using closed-form expressions for the first/second derivatives, as follows  
\vspace{-0.2cm}
\begin{align}
&J(\theta) = \sum_{n=1}^N \sum_{m=0}^{1} 
\gamma_{n,m}^{(t)} \,\big(x_n - \mu_m(\psi,\theta)\big)^2, \\
&\theta^{(t+1)}=\theta^{(t)}-\frac{g(\theta^{(t)})}{h(\theta^{(t)})},
\end{align}
where the first derivative $g(\theta)$ and the second derivative $h(\theta)$ are given by
\vspace{-0.1cm}
\begin{align}
&g(\theta)
=\partial_\theta J(\theta) \nonumber \\
&=2A \sum_{n=1}^N \sum_{m=0}^{1} \gamma_{n,m}^{(t)} 
\Big[ x_n \sin(\theta+c_m) - \tfrac{A}{2}\sin\!\big(2(\theta+c_m)\big)\Big], \label{eq:grad}\\[1ex]
&h(\theta)
=\partial_{\theta}^2 J(\theta) \nonumber \\
&=2A \sum_{n=1}^N \sum_{m=0}^{1} \gamma_{n,m}^{(t)} 
\Big[ x_n \cos(\theta+c_m) - A\cos\!\big(2(\theta+c_m)\big)\Big], \label{eq:hess}
\end{align}
where $c_m \triangleq \varphi_m - \psi$.  Starting from an initial guess $\theta^{(0)}$, 
Newton's method converges rapidly to the maximizer of $\mathcal{Q}(\theta|\theta^{(t)})$. The output of a full EM step is an updated estimate of the unknown rotation parameter, denoted by $\hat{\theta}$, together with a ML (hard) detection of the transmitted symbol (index) sequence, given by
\vspace{-0.25cm}
\begin{equation}
	\hat{s}_n = \arg\max_{m\in\{0,1\}} \gamma_{n,m}^{(t)}.
\end{equation}
Each EM iterative step consists of at most $L_{\max}$ iterations, or it terminates earlier if the difference between two consecutive estimates falls below a tolerance $\varepsilon$. In addition, Newton algorithm runs until a tolerance criterion is satisfied or until a maximum number of iterations $N_{\max}$ is reached.

\subsubsection{Outer loop (tuning of the LO phase $\psi$)}
After estimating the channel phase $\hat{\theta}$ from the current block via the EM algorithm (inner loop),  the receiver updates the LO phase $\psi$ in an iterative manner. 
The guiding principle (see Proposition \ref{p2}) is that sensing performance is maximized when the effective offset 
$\phi=\hat{\theta}-\psi$ is close to quadrature ($\phi\simeq (2k+1)\pi/2$), whereas communication performance is favored when $\phi$ is shifted towards ($\phi\simeq k\pi$). 
Based on this observation, we set a sensing-target $\psi_{\text{sen}}=\hat{\theta}+\frac{\pi}{2}(2k+1)$, and a communication-target $\psi_{\text{com}}=\hat{\theta}+k\pi$. 

At each outer iteration, the Fisher-information constraint in [P1] is evaluated using the current 
$(\psi,\hat{\theta})$ by using \eqref{fish}. If the constraint is violated, the update moves $\psi$ toward $\psi_{\text{sen}}$; 
otherwise, it moves toward $\psi_{\text{com}}$. To avoid large jumps and control fluctuation, the update is performed along a granular angular displacement within the fundamental sector $[0,\pi]$. Specifically, the iterative rule is given by
\vspace{-0.1cm}
\begin{equation}
\psi_{t+1} = \psi_t + \lambda \,\mathrm{wrap}_\pi\big(\psi_{\text{tar}}-\psi_t\big)\;\;\;(\!\!\!\!\!\!\mod \pi),
\end{equation}
where $\psi_{\text{tar}}\in\{\psi_{\text{sen}},\psi_{\text{com}}\}$ depending on feasibility (eq. \eqref{fish}), 
$\lambda\in (0,1]$ is an auxiliary step parameter, and $\mathrm{wrap}_\pi(x)=x-\pi\operatorname{round}(x/\pi)$ adjusts $\psi_{\text{tar}}$ to avoid large jumps and numerical fluctuations in the LO phase update. Once the new $\psi$ is set, the algorithm returns to the inner EM loop. The basic steps of the QISAC algorithm are summarized in Algorithm~\ref{alg1}. It is worth noting that, by directly maximizing the Fisher information in \eqref{fish}, the optimal sensing offset can be obtained for any SNR. To avoid SNR-dependent offsets, we adopt the high-SNR-optimal sensing angle, used for analytical insight, which admits a closed-form expression and remains nearly optimal in the moderate-to-high SNR regime, since the gradual $\lambda$-based update only steers the LO toward a sensing-friendly direction rather than relying on an exact SNR-specific angle.

\noindent {\it Computational complexity:} The proposed QISAC scheme uses an iterative inner--outer structure. The inner loop runs the EM algorithm with complexity $\mathcal{O}(N I_{\text{EM}})$ per block, where $I_{\text{EM}} \leq L_{\max}$. The Newton update and LO retuning in the outer loop incur negligible cost. In contrast, a brute-force joint search over all symbol sequences and LO phases grows exponentially with $N$ and is computationally prohibitive, making the proposed approach efficient and scalable compared with exhaustive methods.

\begin{algorithm}[t]\label{alg1}
	\small
	\caption{QISAC scheme for homodyne BPSK}
	\KwIn{$\{x_n\}_{n=1}^N$, $\psi^{(0)}$, $\sigma^2$, $A$, $\Gamma_{\min}$, $\lambda$, $\varepsilon$, $T_{\max}$, $L_{\max}$, $N_{\max}$.}
	\KwOut{$\hat\theta$, $\hat\psi$, $\{\hat s_n\}$}
	$t\!\leftarrow\!0$; $\psi\!\leftarrow\!\psi^{(0)}$ \;
	\While{$t<T_{\max}$}{
		\tcp{EM for fixed $\psi$}
		initialize $\theta$; \Repeat{$|\Delta\hat\theta|<\varepsilon$ \rm{or max. number of iterations $L_{\max}$}}{
			\textbf{E:} $\gamma_{n,m}\!\propto\!\exp\!\big(-\tfrac{(x_n-A\cos(\theta+\varphi_m-\psi))^2}{2\sigma^2}\big)$, $m\!\in\!\{0,1\}$; normalize over $m$\;
			\textbf{M:} Newton method for $\theta$ using weighted quadratic cost; $\hat\theta\!\leftarrow\!$ update\;
		}
		$\theta\!\leftarrow\!\hat\theta$; and $\hat s_n\!=\!\arg\max_m\gamma_{n,m}$\;
		\tcp{Choose target $\psi$ by feasibility}
		$F_c\!\leftarrow\!NF(\psi,\hat\theta)$ by using \eqref{fish}\;
		$\psi_{\text{com}}\!\leftarrow\!\hat\theta+k\pi$;\quad
		$\psi_{\text{sen}}\!\leftarrow\!\hat\theta+\tfrac{\pi}{2}(2k+1)$\;
		$\psi_{\text{tar}}\!\leftarrow\!
		\begin{cases}
			\psi_{\text{com}}, & F_c\ge\Gamma_{\min}\\
			\psi_{\text{sen}}, & \text{otherwise}
		\end{cases}$\;
		$\Delta\psi\!\leftarrow\!\mathrm{wrap}_{\pi}(\psi_{\text{tar}}-\psi)$;\quad
		$\psi\!\leftarrow\!\psi+\lambda\,\Delta\psi\;\;(\!\!\!\mod \pi)$\;
		\If{$|\Delta\psi|<\varepsilon$}{\textbf{break}}
		$t\!\leftarrow\!t+1$\;
	}
	\Return{$\hat\theta\!=\!\theta$, $\hat\psi\!=\!\psi$, $\{\hat s_n\}$}
\end{algorithm}

\begin{figure}
\centering
	\includegraphics[width=0.87\linewidth]{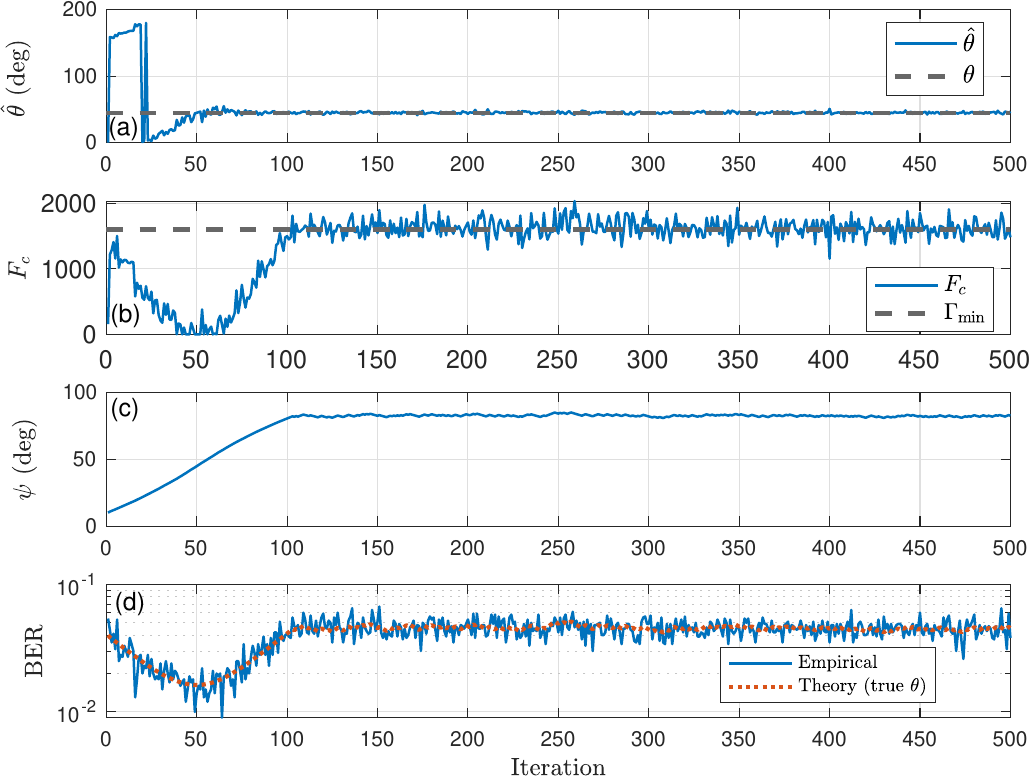}
    \vspace{-0.3cm}
	\caption{QISAC performance versus the number of iterations. 
		(a) Estimated $\hat{\theta}$ (dashed: true $\theta=45^\circ$); 
		(b) Fisher information (dashed: $\Gamma_{\min}=0.6F_c^{\max}$); 
		(c) LO phase $\psi$; 
		(d) BER: empirical (dashed: theory with $\phi=\theta-\hat{\psi}$).
		Setting: $N=1000$, and $\theta=45^\circ$.}
	\label{fig1}
\end{figure}

\begin{figure}
\centering
	\includegraphics[width=0.87\linewidth]{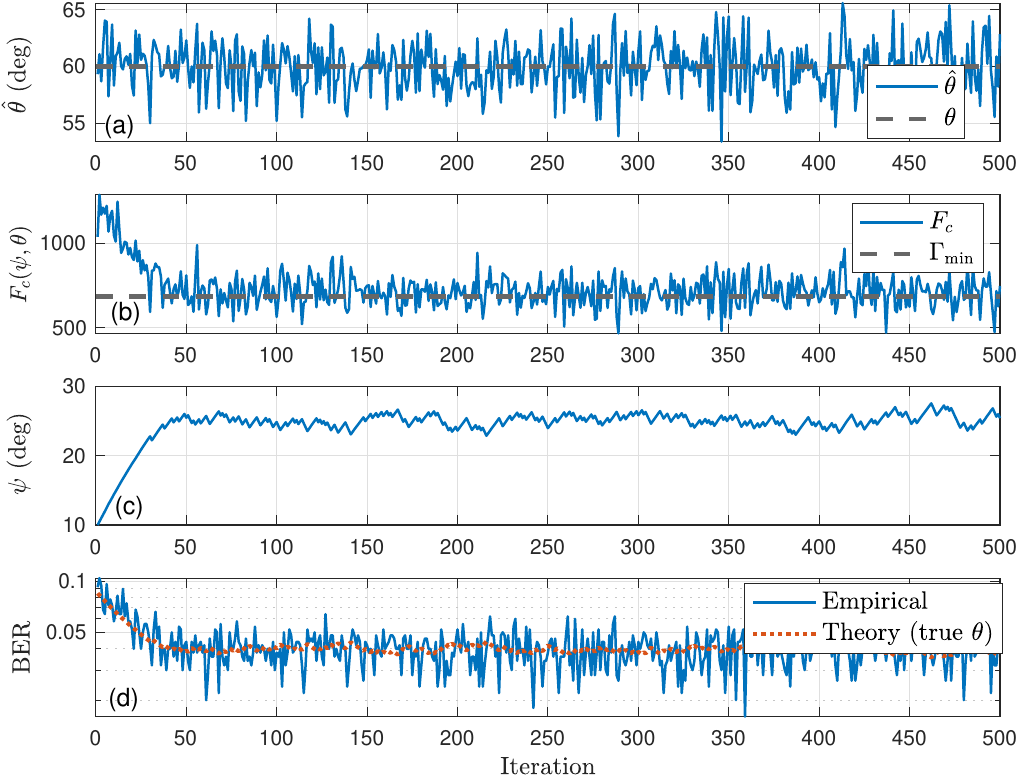}
	 \vspace{-0.3cm}
		\caption{QISAC performance versus the number of iterations. 
		(a) Estimated $\hat{\theta}$ (dashed: true $\theta=60^\circ$); 
		(b) Fisher information (dashed: $\Gamma_{\min}=0.
		5F_c^{\max}$); 
		(c) LO phase $\psi$; 
		(d) BER: empirical (dashed: theory with $\phi=\theta-\hat{\psi}$).
		Setting: $N=500$, and $\theta=60^\circ$.}
	\label{fig2}
\end{figure}

\vspace{-0.36cm}
\section{Numerical results}

We validate the proposed QISAC scheme via computer simulations. Unless stated otherwise, we adopt the common setting with $E=10$, $\eta=0.8$, $\lambda=0.01$, $N_a=3$, $\varepsilon=10^{-3}$ (tolerance for $\psi$ (outer), EM and Newton (inner) updates), $T_{\max}=L_{\max}=500$ iterations, and $N_{\max}=100$ iterations.

Fig. \ref{fig1}(a--d) illustrates the rapid convergence of the proposed QISAC algorithm for a representative configuration with $\theta=45^\circ$, $N=1000$ symbols, $\Gamma_{\min}=0.6F_c^{\max}$. Our key observation: (a) The estimated phase $\hat{\theta}$ locks to the true value within a few outer iterations and thereafter exhibits only small residual jitter (approximately $\pm1^\circ$), indicating stable behavior of the EM inner loop under the chosen operating point. (b) The block Fisher information $F_c$ quickly rises to the prescribed threshold $\Gamma_{\min}$ and remains close to it, with fluctuations attributable to finite-sample noise. (c) Correspondingly, the LO phase $\psi$ converges to a fixed operating value $\approx 85^\circ$ after roughly $100$ outer iterations, reflecting the algorithm's retuning toward an appropriate sensing-communication compromise. (d) The empirical BER also settles to a stable level and closely matches the theoretical benchmark computed for the instantaneous estimate $\phi=\theta-\hat{\psi}$, confirming the validity of the approximate expression  in our operating regime.

Fig.~\ref{fig2}(a--d) illustrates the same algorithm under a smaller block
size, with $\theta=60^\circ$, $N=500$, and $\Gamma_{\min}=0.5F_{c}^{\max}$.
As expected, the reduced $N$ increases finite-sample variability; in
particular, both $\hat{\theta}$ and $F_c$ exhibit visibly larger
fluctuations, and the steady-state BER shows higher variance than in
Fig.~\ref{fig1}. The looser Fisher-information constraint biases the LO
toward a more communication-favorable operating point (smaller
$|\phi|$). Nevertheless, because of the smaller $N$, the steady-state
BER remains comparable to that of Fig.~\ref{fig1}.

Finally, Fig. \ref{fig3} presents the steady-state communication-sensing trade-off (BER versus normalized block Fisher $\Gamma_{\min}/F_{c}^{\max}$). The dashed curves show the theoretical performance for known $\theta$, while markers report Monte-Carlo results for the QISAC scheme. A clear Pareto trade-off appears {\it i.e.}, increasing the Fisher requirement drives $|\phi|$ toward $\pi/2$, improving estimation accuracy at the cost of a larger BER; relaxing the Fisher requirement moves $\phi$ toward $0$, minimizing BER. Consistent with the benchmark scaling $F_c^{\max}=N A^2/\sigma^2$, larger $N$ or smaller noise variance ``enlarge'' the achievable region, since the sensing requirement can be met with a more communication-friendly phase, yielding lower BER. 

\begin{figure}
\centering
	\includegraphics[width=0.75\linewidth]{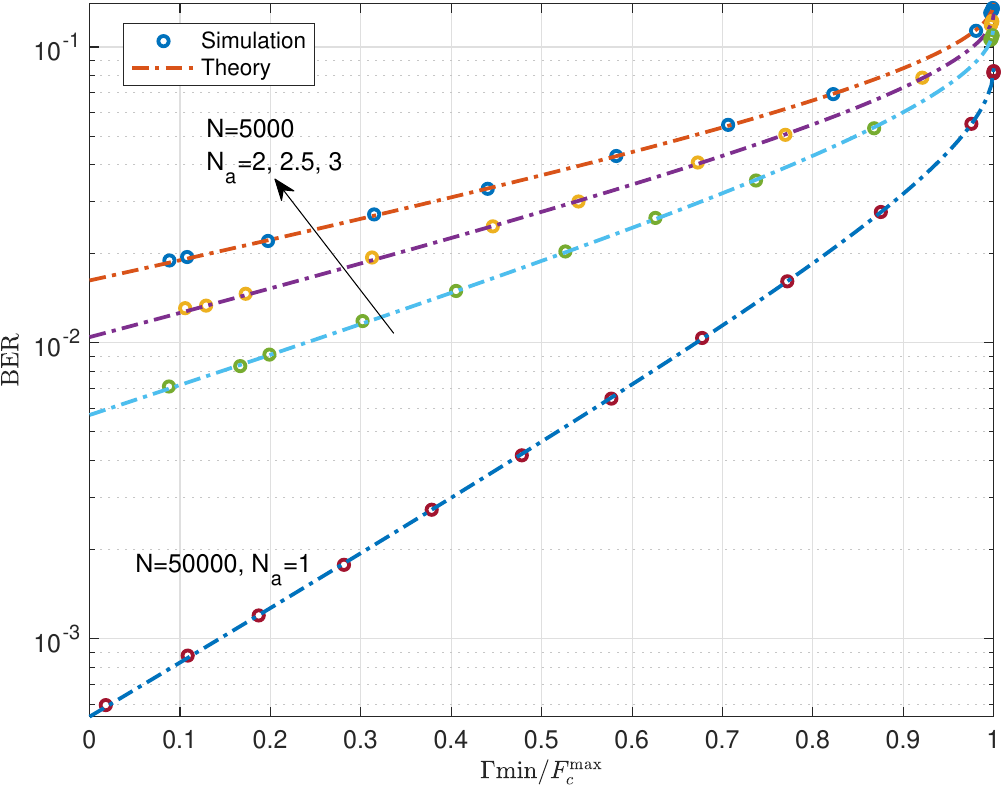}
	\vspace{-0.3cm}
\caption{Tradeoff between communication performance (BER) and normalized estimation/sensing accuracy ($\Gamma_{\min}/F_c^{\max}$); simulation results (markers), theoretical results corresponding to known $\theta$ (dashed lines) . Setting: $\theta = 30^\circ$, $N_a \in \{1, 2,\,2.5,\,3\}$, $N \in \{5000,\,50000\}$.}
	\label{fig3}
\end{figure}

\vspace{-0.2cm}
\section{Conclusion}

In this letter, we have studied a practical QISAC framework for quantum optical links with BPSK modulation and homodyne detection. By formulating the design as a BER minimization problem under a Fisher-information-based constraint, we developed a low-complexity iterative algorithm that jointly estimates the unknown channel rotation and detects the transmitted symbols. Numerical results validated the effectiveness of the proposed approach and demonstrated the fundamental trade-off between communication performance and sensing quality. Future extensions will investigate non-BPSK modulation schemes, which may offer improved efficiency and broader applicability, as well as non-classical light sources such as squeezed or entangled states to surpass the standard quantum limit within the proposed QISAC framework.

\appendices

\vspace{-0.3cm}
\section{Proof of proposition 1}\label{ap1}

For BPSK under homodyne detection with LO phase $\psi$ and effective offset
$\phi$, the Euclidean distance between the symbols is written as 
\vspace{-0.1cm}
\begin{align}
	d_{\min}(\phi)
	&=\left|A\cos(\varphi_0+\phi)
	- A\cos(\varphi_{1}+\phi)\right| \nonumber\\
	&=2A\left|\sin\!\big(\phi-\tfrac{\pi}{2}\big)\right|=2A|\cos(\phi)|. \label{ll}
	\end{align}
With Gaussian noise of variance $\sigma^2$, the BER is equal to 
\vspace{-0.1cm}
\begin{align}
	&P_e(\psi,\theta)=Q\!\left(\frac{d_{\min}(\phi)}{2\sigma}\right)=Q\!\left(
	\frac{A|\cos(\phi)|}{\sigma}.
	\right).
\end{align}
For the Fisher information, the likelihood is a Gaussian mixture
\vspace{-0.2cm}
\begin{align}
	p(x;\psi, \theta)=\frac{1}{2}\sum_{m=0}^{1}
	\mathcal{N}\!\big(x;\mu_m,\sigma^2\big).
\end{align}
The derivative of the log-likelihood with respect to $\theta$ is
\vspace{-0.1cm}
\begin{align}
&\partial_\theta \log p(x;\psi,\theta)\!=\! \frac{\sum_{m=0}^{1} 
        \mathcal{N}\!\big(x;\mu_m,\sigma^2\big)\,
        (x-\mu_m)\,\mu'_m}
        {\sigma^2 \sum_{m=0}^{1} 
        \mathcal{N}\!\big(x;\mu_m,\sigma^2\big)}.
\label{eq:deriv1} 
\end{align}
Therefore, the symbol Fisher information is
\vspace{-0.1cm}
\begin{align}
&F(\psi,\theta)
=\int_{-\infty}^{\infty} \big(\partial_\theta \log p(x;\psi,\theta)\big)^2\,p(x;\psi,\theta)\,dx \nonumber\\
&=\frac{1}{2}\int_{-\infty}^{\infty}\!\!\!
\frac{\Big(\sum_{m=0}^{1}\mathcal{N}\!\big(x;\mu_m,\sigma^2\big)\,
\frac{(x-\mu_m)\,\mu'_m}{\sigma^2}\Big)^2}
{\sum_{m=0}^{1}\mathcal{N}\!\big(x;\mu_m,\sigma^2\big)}\,dx.
\end{align}

\vspace{-0.2cm}
\section{Proof of proposition 2}\label{ap2}

For the communication task, the error probability in \eqref{error} decreases as the argument of the $Q$-function increases, which is proportional to $|\cos(\phi)|$. The maximum of the cosine term occurs at $\phi=k\pi$, which yields the communication-optimal phase
$\psi_{\mathrm{com}}=\theta+k\pi$, with $k\in \mathbb{Z}$.

For the estimation/sensing task, we consider the Fisher information in~\eqref{fish}, defined as the
expectation of the squared derivative of the log-likelihood with respect to $\theta$.
The received distribution is a Gaussian mixture with means
$\mu_m=A\cos(\varphi_m+\phi)$ and derivatives
$\mu'_m=-A\sin(\varphi_m+\phi)$. At finite SNR, the derivative contains
cross terms between mixture components, but in the high-SNR regime
($A/\sigma\to\infty$) these terms vanish because the Gaussian lobes are well
separated and each observation effectively originates from a single symbol.
In this regime the symbol Fisher reduces to the average of the per-symbol Fisher values
\vspace{-0.1cm}
\begin{align}
F(\psi,\theta) &\approx \frac{1}{2}\sum_{m=0}^{1}\frac{(\mu'_m)^2}{\sigma^2}
= \frac{A^2}{2\sigma^2}\left(\sin^2(\phi)+\sin^2(\pi+\phi) \right) \nonumber\\
&=\frac{A^2}{\sigma^2}\sin^2(\phi).
\end{align}
Based on the above, the Fisher information is maximized for $\sin^2(\phi)=1\Rightarrow\psi_{\mathrm{sen}}=\theta+\frac{\pi}{2}(2k+1)$,  with $k\in \mathbb{Z}$.

\end{document}